\documentclass{svjourpab}

\usepackage{graphicx,color,colortbl}
\usepackage{tikz}
\usepackage{amsmath,mathrsfs,amsfonts,latexsym,amssymb,verbatim}
\usepackage[mathcal]{euscript}

\renewcommand{\t}{\textsc{t}}
\newcommand{\f}{\textsc{f}}

\newcommand{\couic}[1]{}
\newcommand{\couicfootnote}[1]{}
\newcommand{\couicefootnote}[1]{}

\newcommand{\ket}[1]{| #1 \rangle}
\newcommand{\bra}[1]{\langle #1 |}

\newcommand{\trace}{\textrm{Tr}}

\newcommand{\ii}{\mathrm i}

\renewcommand{\H}{\mathcal{H}}

\newcommand{\N}{\mathcal{N}}
\newcommand{\Z}{\mathbb{Z}}

\begin{document}

\title{An overview of Quantum Cellular Automata} 
\subtitle{}
\titlerunning{QCA---overview}

\author{P. Arrighi}

\institute{
P. Arrighi \at Aix-Marseille Univ., Universit\'e de Toulon, CNRS, LIS, Marseille, and IXXI, Lyon, France\\
\email{pablo.arrighi@univ-amu.fr} 
}

\date{\today}
\maketitle

\begin{abstract}
Quantum cellular automata are arrays of identical finite-dimensional quantum systems, evolving in discrete-time steps by iterating a unitary operator $G$. Moreover the global evolution $G$ is required to be causal (it propagates information at a bounded speed) and translation-invariant (it acts everywhere the same). Quantum cellular automata provide a model/architecture for distributed quantum computation. More generally, they encompass most of discrete-space discrete-time quantum theory. We give an overview of their theory, with particular focus on structure results; computability and universality results; and quantum simulation results.  
\end{abstract}

\vspace{0.7cm}
{\centering {\em This paper is dedicated to my high school Mathematics teacher, Anne Lef\`evre.}}

\section{Introduction}

Von Neumann provided the modern axiomatisation of quantum theory in terms of the density matrix formalism{ \cite{NeumannQT}} in 1955. He also invented the cellular automata (CA) model of computation \cite{Neumann} in 1966, but never brought the two together. Feynman suggested doing so \cite{FeynmanQC,FeynmanQCA} in 1986, just as he was inventing the very concept of quantum computation (QC).

Indeed, confronted with the inefficiency of classical computers for simulating quantum physics, Feynman realized that one ought to use quantum devices instead \cite{FeynmanQC}. What better than a quantum system in order to simulate another quantum system? Soon afterwards \cite{FeynmanQCA} he introduced Quantum Cellular Automata (QCA) for two reasons. First because they constituted a promising architecture for the implementation of quantum simulation devices---as demonstrated nowadays with cold atoms on optical lattices, integrated quantum optics or superconducting qubits. Second, because the quantum simulation of a quantum physical phenomena requires that we are able to describe it ``in terms of qubits''. Most often, this qubit description is obtained by formulating a discrete-space discrete-time version of the original continuous description of the phenomena---i.e. a QCA model for it. 

Notice that your usual, numerical simulation of the phenomena on a classical computer, would also require that the phenomena be described ``in terms of bits''. But these numerical schemes are not usually thought of as being physically legitimate themselves, because they tend to be unaesthetic or worse break fundamental symmetries. For instance, applying finite-difference methods upon the partial differential equation governing the propagation of a particle, will typically break unitarity, making it unphysical. A QCA model, on the other hand, has to remain physical, and unitary (with some work it may even retain Lorentz-covariance). In this sense, QCA models may be thought of as constituting physically legitimate descriptions of the phenomena themselves. Moreover some are way simpler and more explanatory than the original continuous description, as we shall see. Thus, the provision of toy models for theoretical physics is another, strong reason to study QCA.

Yet another strong reason lies at the heart of Theoretical Computer Science with the basic question: What is a computer, ultimately? Which key resources are granted to us by nature, for the sake of computing? Tentative answers to these questions have been obtained by abstracting away from particles and forces, to reach formal models of computation---such as the Turing machine. Turing machines used to be our best answer. Nowadays, however, spatial and quantum parallelism need be taken into account, leading us to propose models of distributed quantum computation. Amongst the different models of distributed QC, QCA are of the most established. Just like CA, QCA account for space ``as we know it'' (i.e. mostly euclidean). Thus, they constitute a framework to model and reason about problems in spatially-sensitive distributed Quantum Computation \cite{MazoyerFiring}. For instance we may wonder, as we shall see, whether there exists an intrinsically universal QCA, i.e. capable of simulating all others in a spacetime-preserving manner; or whether QCA evolutions are computable altogether. 

\begin{figure}
\centering
\includegraphics[scale=0.5, clip=true, trim=0cm 0cm 0cm 0cm]{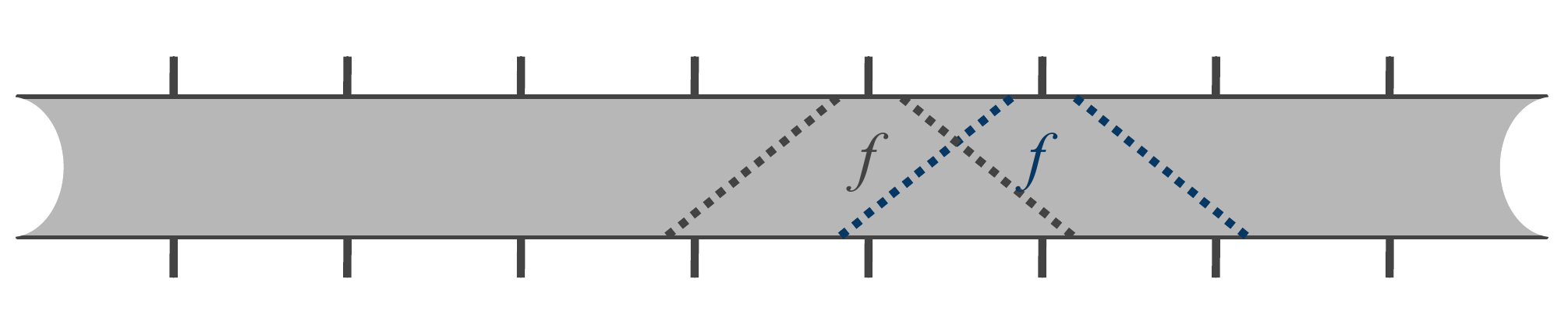}
\caption{\label{fig:structure} \label{fig:AxiomaticQCA} {\em Axiomatic definition of a QCA} : a translation-invariant unitary operator, such that information does not propagate too fast. Each wire carries a quantum system, time flows upwards.}
\end{figure}

A QCA is an array of identical finite-dimensional quantum systems. The whole array evolves in discrete-time steps by iterating a linear operator $G$. Moreover the global evolution $G$ is required to be translation-invariant (it acts everywhere in the same way), causal (information cannot be transmitted faster than some fixed number of cells per time step), and unitary (the condition required by the postulate of evolutions in quantum theory, akin to reversibility). See Fig. \ref{fig:AxiomaticQCA}. This style of definition is `axiomatic', in the sense that it characterizes QCA as the sole mathematical object fulfilling a number of high-level principles. It is the natural `quantization' of the classical definition \cite{SchumacherWerner,ArrighiLATA,ArrighiIJUC}. But contrary to its classical counterpart the axiomatic definition did not immediately yield a straightforward way of constructing the instances of this model. A great deal of effort has been dedicated towards understanding their structure, in terms of infinitely repeating circuits of local, quantum gates  \cite{ArrighiUCAUSAL,ArrighiJCSS,ArrighiPQCA}.

\paragraph{Roadmap.} We will tackle the above themes in a somewhat reversed order. We will start with the axiomatic definition of QCA, the consequent structure theorems, their origins, in Section \ref{sec:structure}. This will place us in a position to recall the main universality results and consequences in computability, in Section \ref{sec:universality}. QCA models of particle physics, whether for the sake of quantum simulation or as toy models for theoretical physics, will be discussed in Section \ref{sec:simulation}. Section \ref{sec:conclusion} summarizes the progress QCA theory has made, and some of the challenges that remain ahead. 

\paragraph{Foreword.} I aimed to state the fundamental results of the field in a mathematical manner, and then touch on the many fascinating results around with just a word of explanation---hoping to show how they relate to each other and thereby draw a coherent picture. I favoured the logical order over the chronological, and citations over long and necessarily incomplete lists of surnames.\\ 
Even with more than $120$ citations and in spite of my best efforts, I realize that this remains to some extent a personal account: I needed to select what seemed to be the most significant contributions, and may still be unaware of other great advances. Do get in touch if you have suggestions for the subsequent versions. Whist writing I learnt that a technically more comprehensive review on QCA was being written \cite{FarrellyReview}. Hopefully the two will complement each other, updating \cite{Wiesner}.\\
I needed to assume a certain knowledge of both the foundations of CA theory, and the foundations of quantum theory. Two great references for this purpose are \cite{KariSurvey} and \cite{NielsenChuang}, respectively.  

\section{Structure}\label{sec:structure}

\subsection{State space}\label{subsec:statespace}

A Quantum Cellular Automata (QCA) is an $n$-dimensional array of identical $d$-dimensional quantum systems. In other words, each cell is a qudit, i.e. a normalised vector in the $d$-dimensional complex space $\mathbb{C}^d$. Since there are $\Z^n$ such cells, the overall state space should morally be something like ``$\mathcal{H}=\bigotimes_{\mathbb{Z}} \mathbb{C}^d$''. Unfortunately this is not a Hilbert space (e.g. the inner product may diverge), see e.g. \cite{NielsenChuang} Subsections 2.1--2.2 for recaps on the notions of inner products, Hilbert spaces and tensor products. If one is willing to switch from the Schr\"odinger picture to Heisenberg picture and pay the price of abandoning Hilbert spaces for $C^*$-algebras \cite{BratteliRobinson}, then one can make sense of QCA over such a space \cite{SchumacherWerner}. In \cite{ArrighiLATA,ArrighiIJUC} we were able to develop a simpler alternative, which is to assume that basic configurations are mostly empty. The equivalence between the two approaches is given in \cite{ShakeelInfinite}. 
\begin{definition}[Configurations]
Consider $\Sigma$ a finite set, called the {\em alphabet}, with $0$ a distinguished element of $\Sigma$, called the \emph{empty} state. A \emph{configuration} $c$ over $\Sigma$ is a function $c:
\Z^n \longrightarrow \Sigma$, i.e. mapping $(i_i,\ldots,i_n)$ to $c_{i_i\ldots i_n}\in\Sigma$, 
such that the set of the $(i_i,\ldots,i_n)\in\Z^n$ such that $c_{i_i\ldots i_n}\neq 0$, is finite. 
The set of all configurations will be denoted $\mathcal{C}_{\Sigma}$ or just $\mathcal{C}$.
\end{definition}
Notice that $\mathcal{C}$ is countable. Thus, we can now consider the Hilbert space of superpositions of configurations. 
\begin{definition}[State space] 
The {\em Hilbert space of configurations} is that having orthonormal basis $\{\ket{c}\}_{c\in\mathcal{C}}$. It will be denoted $\mathcal{H}_{\mathcal{C}}$ or just $\mathcal{H}$.
\end{definition}

\subsection{QCA}\label{subsec:QCA}

The global evolution of a QCA is required to be translation-invariant, meaning that it acts everywhere in the same way.
\begin{definition}\textbf{(Translation-invariance)}\label{def:translation-invariance} 
Let $\tau_k$ denote the translation operator along the $k^{{th}}$ dimension, i.e. the linear operator over $\mathcal{H}$ which maps $\ket{c}$ into $\ket{c'}$, where $\ket{c'}$ is such that for all $(i_1,\ldots,i_n)$, $c'_{i_1\ldots i_k \ldots i_n}=c_{i_1\ldots i_k+1 \ldots i_n}$. A linear operator $G$ over $\mathcal{H}$ is said to be 
\emph{translation-invariant} if and only if $G\tau_k=\tau_k G$ for every $k$.
\end{definition}
Moreover, the global evolution of a QCA is required to be causal, meaning that information propagates at a bounded speed. In order to formulate this property, we need to be able to speak of the state of a cell $x$ at time $t+1$, to say that it should only depend on the state of its neighbours $x+\N$ at time $t$---where $\N$ is a fixed set of vectors, which added to any $x$, lead to its neighbours. But in order to speak of the state of a subsystem in quantum theory we must switch to the density matrix formalism, which is not so trivial---again see e.g. \cite{NielsenChuang} Subsection 2.4. Summarizing, a density matrix represents a probability distributions over pure states $\{p_i, \ket{\psi_i}\}$ as the corresponding convex sum of projectors $\rho=\sum_i p_i \ket{\psi_i}\bra{\psi_i}$. Thus, when pure states $\ket{\psi}$ evolve according to $\ket{\psi'}=G\ket{\psi}$, density matrices $\rho$ evolve according to $\rho'=G\rho G^\dagger$. Then, the state of cell $x=(i_1,\ldots,i_n)$ at time $t+1$ is obtained by tracing out all the of the other cells, i.e. $\rho'_x=\trace_{\overline{x}}(\rho')$, with $\trace_{\overline{S}}(.)$ the linear operator such that $\trace_{\overline{S}}(\ket{c}\bra{d})=(\delta_{c_{\overline{S}},d_{\overline{S}}})\ket{c_S}\bra{d_S}$. Similarly, the state of its neighbours $x+\N$ at time $t$ is obtained by tracing out the rest, i.e. $\rho_{x+\N}=\trace_{\overline{x+\N}}(\rho)$. We get :
\begin{definition}[Causality]\label{def:causality} 
A linear operator $G$ over $\mathcal{H}$ is said to be 
\emph{causal} with neighbourhood $\N\subset \Z_n$ if and only if for any $x\in\Z_n$ there exists a function $f$ such that for any $\rho$ over $\mathcal{H}$, we have $\rho'_x = f (\rho_{x+\N}),$ where $\rho'=G\rho G^\dagger$.
\end{definition}
To a certain extent, this $f$ may be thought of as the equivalent of the local rule of a classical CA. However, unlike for classical CA, this $f$ does not straightforwardly yield a local mechanism whereby $\rho'$ may be computed from $\rho$, for two reasons. First, because $f$ by itself is not unitary (it maps many cells into one) and thus cannot be considered to be physical. Second, because in quantum theory, knowing the states $\rho'_x$ and $\rho'_{y}$ of cells $x$ and $y$ does not entail knowing their joint states $\rho'_{\{x,y\}}$---as these may be entangled. Hence, unlike for classical CA, the following axiomatic definition of QCA does not immediately yield a straightforward way of constructing / enumerating all of the instances of the model.
\begin{definition}[QCA]\label{def:qca} 
A {\em QCA} is a linear operator over $\mathcal{H}$
which is translation-invariant, causal and unitary.
\end{definition}
We needed structure theorems in order to tame this axiomatic definition, into a constructive one.

\subsection{Structure theorems}\label{subsec:structh}

From the above axiomatic definition of QCA, we were able to eventually deduce that every QCA can be directly simulated by a finite depth quantum circuit of local unitary gates, infinitely repeating across space \cite{ArrighiUCAUSAL,ArrighiJCSS}. In order to do so each cell of the QCA needs be encoded into a doubled up cell.
\begin{theorem}[Unitary plus causality implies localizability]\cite{ArrighiUCAUSAL,ArrighiJCSS}\label{th:ucausal}~\\
Let $G$ be an $n$-dimensional QCA with alphabet $\Sigma$. Let $E_x$ be the isometry from $\H_\Sigma$ to $\H_\Sigma\otimes\H_\Sigma$ which adds an ancillary empty subcell at $x$, i.e.  $E_x\ket{\psi}=\ket{0}\otimes\ket{\psi}$. This mapping can be trivially extended to whole configurations, yielding the mapping $E:\H_{C^{\Sigma}}\to\H_{C^{\Sigma^2}}$. There exists an $n$-dimensional QCA $H$ with alphabet $\Sigma^2$, such that $HE=EG$, where $H$ admits the following multi-layer quantum circuit representation:
\begin{align*}
H=(\bigotimes S_x)(\prod_x K_x) \label{eq:Kprod}
\end{align*}
with
\begin{itemize}
\item $S_x$ is the swap between the two subcells at $x$ and hence is local to $x$.
\item $K_x$ is $(G^\dagger S_x G)$, which turns out to be local to the neighbourhood $x+\N$.
\end{itemize}
\end{theorem}
\begin{proof}{\em outline.}
\begin{enumerate}
\item Show the equivalence between causality in the Schr\"odinger picture (Def. \ref{def:causality}) and causality in the Heisenberg picture \cite{SchumacherWerner}, which states that if $U$ is causal with neighbourhood $\N$ and $A$ of a local operator $A_x\otimes I$, then $A'=U^\dagger AU$ is a local operator $A'_{x+\N}\otimes I$. 
\item Extend $G$ to act only on the right subcells, leaving the left subcells unchanged. Apply the previous point to obtain that each $K_x$ is a local operator $K_{x+\N}\otimes I$.
\item Notice that $\prod_x K_x=\ldots G^\dagger S_x G G^\dagger S_{x+1} G\ldots = G^\dagger (\prod_x S_x) G$. Thus $\prod_x K_x$ computes $G$ and swaps it away to the left subcells.
\end{enumerate}
\end{proof}
Remark that each $K_x$ plays the role a local update mechanism for the cell $x$, as it computes the future state of that cell from the local information available at $x+\N$, and then puts it aside in the left subcell. It may alter cells $x+\N$ for this purpose, but this is not a problem because the $K_x$ commute with one another.

This `direct simulation' of QCA $G$ by QCA $H$ does show that any QCA can be put into the form of a finite depth quantum circuit, with local unitary gates $K$ and $S$, up to a simple encoding. Indeed, $K_x$ and $K_y$ can be done in parallel whenever they do not overlap, i.e. when $x+\N\ \cap\ y+\N =\varnothing$. In the $\N=\{0,1\}$ case, for instance, the construction results in a $2n$--layered circuit. This $\N=\{0,1\}$ case may seem restrictive, but it is not. Just like for CA \cite{IbarraJiang}, one can always group the cells into supercells, relative to which the neighbourhood reduces down to $\N=\{0,1\}$.\\ 
The notion of `intrinsic simulation' of a QCA $G$ by a QCA $H$ is obtained by relaxing the notion of direct simulation in two ways. First, by allowing for such groupings of cells into supercells, yielding $G'$ and $H'$. Second, by allowing $H'$ to be iterated $k$ times before it directly simulates $G'$. I.e. $H$ intrinsically simulates $G$ when $H'^k$ directly simulates $G'$. The notion, made rigourous in \cite{ArrighiNUQCA}, allowed us to reach an even simpler quantum circuit-form for QCA \cite{ArrighiPQCA}. 
\begin{definition}[PQCA]\label{def:pqca}
An $n$-dimensional {\em partitioned  QCA} (PQCA) $G$ is induced by a {\em scattering unitary} $U$ taking a hypercube of $2^n$ cells into a hypercube of $2^n$ cells, i.e. acting over $\mathcal{H}_{\Sigma}^{\otimes 2^n}$, and preserving quiescence, i.e. $U\ket{0\ldots 0}=\ket{0\ldots 0}$. Let $J=(\bigotimes_{2\mathbb{Z}^n} U)$ over $\mathcal{H}$ and $\tau=\tau_1\ldots\tau_n$ the diagonal translation. The induced global evolution is $J$ at even steps, and $\tau^\dagger J \tau$ at odd steps.
\end{definition}
In other words, PQCA work by partitioning the grid of cells into supercells, applying a local operation $U$ on each supercell, translating the partition along $(1 \ldots 1)^T$, applying $U$ to the new macrocells, etc., as illustrated in Fig. \ref{fig:UsimV}. Of course one can now change scale and take the view that supercells are now cells, whose subcells are the former cells. Then the data contained in each cell can be thought of as being subdivided into $2^n$ subcells, about to be sent towards the $(\pm 1 \ldots \pm 1)^T$ direction. In this picture, a PQCA is therefore the alternation of a `scattering step' (where $U$ gets applied on each cell) followed by an `advection step' (where subcells get synchonously exchanged with their corresponding cross-diagonal neighbour)---a scheme which physicists refer to as Lattice Gas Automata. Either way, PQCA suffer the small inconvenience that homogeneity is now over two steps. For instance, in the Lattice Gas Automata picture, the cells at time $t+1/2$ are translated by $(1/2 \ldots 1/2)^T$ with respect to the cells at time $t$. This can be fixed elegantly by doubling up the lattice like a checkerboard, with black and white lattices ignoring each other. Then PQCA recover full-translation invariance and fall back within the class of QCA with neighbourhood $\N=\{-1,+1\}\times\ldots\times\{-1,+1\}$, i.e. QCA such that the state of a cell depends only upon that of its diagonal neighbours at the previous time step---but not on its own previous state, see Fig. \ref{fig:DiracQCA}. We were able to prove the following. 
\begin{theorem}[PQCA are intrinsically universal]\cite{ArrighiPQCA} \label{th:pqca}
Given any $n$-dimensional QCA $G$, there exists an $n$-dimensional PQCA $H$ which intrinsically simulates $G$.
\end{theorem}
\begin{proof}{\em outline.}
\begin{enumerate}
\item Space-group the cells of the QCA $G$ so that its neighbourhood is down to $\N=\{0,1\}\times\ldots\times\{0,1\}$.
\item Apply Th. \ref{th:ucausal}, yielding $2^n$ successive partitions for applying the $K_x$ in a non-overlapping way. 
\item Devise a PQCA $H$, with appropriate ancillary systems used to get the timing right, so that the scattering $U$ successively performs $K$ translated by the vectors in $\N$, according to the corresponding partition. 
\end{enumerate}
\end{proof}

\subsection{Advanced structure theorems}\label{subsec:advancedstructh}

The Theorems of Subsec \ref{subsec:structh} show that QCA can be {\em simulated}, in a spacetime-preserving manner, by multi-layer quantum circuits, and that these can in turn be simulated by PQCA. For RCA, an analogous result was given in \cite{Durand-LoseBlock,ArrighiBLOCKREP}. Still, does it mean that QCA are {\em exactly} multi-layer quantum circuits, or PQCA? A number of researchers have been pursuing this question. Our current knowledge of these issues varies according to the spatial dimension. It helps to recall the corresponding results on Reversible Cellular Automata (RCA), as these are the classical counterparts of QCA. 

In 1D, RCA are exactly the set of translation-invariant two-layer circuits of reversible gates and partial shifts \cite{KariBlock}. The analogous result holds true for 1D QCA, as we can show that they are exactly the set of translation-invariant two-layered quantum circuits \cite{SchumacherWerner,ArrighiLATA,ArrighiIJUC}. This leads to a group theoretical classification of 1D QCA, called the index theory \cite{GrossNesmeVogtsWerner}---a single number characterizes the flux of information within the 1D QCA. Still, 1D QCA are not exactly PQCA \cite{ShakeelMeyer}. 

In 2D, RCA are again exactly the set of translation-invariant three-layer circuits of reversible gates and partial shifts \cite{KariBlock}. This analogous result holds true for, translation-invariance left aside, for 2D QCA  \cite{HastingsClassification,Haah2DExactRep}. Still we were able to show that 2D QCA are not exactly PQCA \cite{ArrighiLATA,ArrighiIJUC}, with a counter-example coming from RCA \cite{KariCircuit}. This is taken a step further in \cite{ShakeelLove}, whose characterization of PQCA in terms of inclusion of algebras actually works independent of the spatial dimension.  

Apart from this characterization of PQCA, little is known in 3D and beyond. It is open whether RCA coincide with the set of translation-invariant circuits of reversible gates and partial shifts \cite{KariBlock}. Again we do not know whether the analogous result holds true for 3D QCA, but \cite{HastingsNonTrivial} holds the promises of a counter-example: ``either a nontrivial three-dimensional qudit QCA exists or a nontrivial two-dimensional fermionic QCA exists.'' Let us take this opportunity to mention fermionic QCA. 

Intuitively, fermions are indistinguishable quantum systems, such that permuting one for another does not change anything but for a global, minus sign. Technically, let $a_x$ denote the linear operator which annihilates the fermion having spin $a$ at $x$, let $a^\dagger_x$ be the corresponding creation, and similarly for the other spins. The fermionic commutation relations are $$\{a_x,a_y\} = \{b_x,b_y\} = \{a_x,b_y\}= \{a_x,b_y^\dagger\} = \{a_x^\dagger,b_y\} = 0,\ \{a_x, a_y^\dagger\} = \{b_x, b_y^\dagger\}=\delta_{x,y}I$$ 
These indeed entail that $a^\dagger_x a^\dagger_y=-a^\dagger_y a^\dagger_x$, so that $a^\dagger_x a^\dagger_x=0$, i.e. two identical fermions exclude each other. However, these anticommutation relations between $a_x$ and $a_y$ for any $x$ and $y$ also entail that $a_x$ is not actually local in the usual, qubit sense. This leads to subtle differences between computing with qubits and computing with fermions \cite{PaviaFermions}. A fermionic QCA is a quantum evolution which is prescribed in terms of how the fermionic operators evolve, rather than how qubit-local algebras evolve, see for instance \cite{PaviaMolecular,HastingsNonTrivial}. However, a direct analog of Theorem \ref{th:ucausal} still holds for fermionic QCA \cite{FarrellyCausal}. Moreover, \cite{FarrellyThesis} shows that fermionic QCA and QCA intrinsically simulate each other. In \cite{ArrighiQED}, the quantum evolution we describe is both a valid QCA and a valid fermionic QCA : whilst the notion of locality differs, that of causality coincides.

Finally, let us mention that Theorem \ref{th:ucausal} cannot be generalized to Noisy QCA, i.e. causal TPCP-maps, as these cannot always be simulated by finite-depth circuits of local TPCP-maps. This is a direct consequence of \cite{ArrighiPCA}, were we showed that there can be no such structure theorems for classical, probabilistic CA. 

\subsection{Classical bijective CA, MPU}\label{subsec:classicalandtensors} 

A reversible CA (RCA) is an invertible CA whose inverse is also causal, i.e. whose inverse is a CA. 
Depending upon the space of configuration which one considers, there may be some invertible CA which are not RCA.
This is the case over ${\cal C}$ in particular, where $F$ causal and invertible does not entail that $F^{-1}$ is causal.
Expectedly when $F^{-1}$ is not causal, its linear extension into a unitary operator $\widehat{F^{-1}}$ over ${\cal H}$ is not causal. Interestingly, however, it then turns out that $\widehat{F}$ is not causal either, even though $F$ was. In other words, the linear extension of a perfectly valid invertible CA $F$ over ${\cal C}$, may lead to a unitary operator $\widehat{F}$ over ${\cal H}$ which fails to be a valid QCA. One way to think about this is that for the update mechanism $K=\widehat{F}^\dagger S  \widehat{F}$ of Th. \ref{th:ucausal} to be local, the local operation $S$ needs be conjugated by causal operators, which may not be the case with $\widehat{F}^\dagger=\widehat{F^{-1}}$. 
We gave a concrete example of this in \cite{ArrighiLATA,ArrighiIJUC}, which is defined as follows. Let $\Sigma=\{0,\t,\f\}$, and for all $a\in\Sigma$ define $+$ as the `exclusive or' $\t+\f=\f+\t=\t$, $a+a=a$, extended so that $a+0=a$, $0+a=0$. Now let $F$ map the configuration $c=\cdots c_{i-1} c_i c_{i+1} \cdots$ into the configuration $F(c)=\ldots (c_{i-1}+ c_i)(c_{i}+c_{i+1})\ldots$.
One sees that $F$ is both bijective over ${\cal C}$ and causal, but that $F^{-1}$ is not causal, because a long subword $\f\f\ldots \f\f$ may either stem from a similar subword $\f\f\ldots \f\f$ or from a subword $\t\t\ldots \t\t$. 
It follows that the corresponding QCA is not causal. Indeed, consider the two states $$\ket{c^\pm}={\textstyle\frac{1}{\sqrt{2}}}\ket{\ldots 00}\otimes \big(\ket{\f\f\ldots \f\f}\pm\ket{\t\t\ldots \t\t}\big)\otimes \ket{00\ldots}$$
and their two images
$$\ket{d^\pm}=\ket{\ldots 00\f\f\ldots \f}\otimes (\frac{\ket{\f}\pm\ket{\t}}{\sqrt{2}}) \otimes \ket{00\ldots}.$$
We can transmit information between arbitrarily distant parties in just one step of $\widehat{F}$ as follows.
\begin{enumerate}
\item Prepare the state $\ket{c^+}$ with the first non quiescent cell in Alice's lab in Paris and the last non quiescent cell with Bob in New York.
\item Alice either leaves the state unchanged or performs a \emph{local change} by applying a phase gate $Z$ to her cell, changing $\ket{c^{+}}$ into $\ket{c^{-}}$.
\item One step of $\widehat{F}$ is performed, leading to either $\ket{d^+}$ or $\ket{d^-}$.
\item Whether Alice performed $Z$ or not has now led to a perfectly measurable change from $\frac{\ket{\f}+\ket{\t}}{\sqrt{2}}$ to $\frac{\ket{\f}-\ket{\t}}{\sqrt{2}}$ for Bob---despite him being arbitrarily far remote.
\end{enumerate}
This infinite speedup is intuitively unphysical, and should disallowed: $\widehat{F}$ is no QCA.\\
Notice that the above example would not work over the space of infinite configurations ${\cal C}_\infty$, because $F$ is non-injective over that space: the infinite configurations $\ldots \f\f \ldots$ and $\ldots \t\t \ldots$ both map to $\ldots \f\f \ldots$. Actually, it turns out that over ${\cal C}_\infty$, any bijective CA is an RCA, because $G$ causal and bijective does entail that $G^{-1}$ is causal. It follows that the linear extension of such a $G$ into a unitary operator $\widehat{G}$ is causal, and the local update mechanism $K=\widehat{G}^\dagger S  \widehat{G}$ is local indeed. Still, the QCA $\widehat{G}$ may have a much wider radius of causality $r_{\widehat{G}}$ than it had as an RCA $G$. In \cite{ArrighiNEIGH} we were able to show that $r_G+r_{G^{-1}}\leq r_{\widehat{G}}=r_{\widehat{G}^\dagger} \leq \min(2r_G+r_{G^{-1}},r_G+2r_{G^{-1}})$. However, given just  $r_G$ one cannot even bound $r_{\widehat{G}}$; in fact there is no computable function $b$ such that $r_{G^{-1}}<b(r_G)$, as was proven in \cite{KariRevUndec}. QCA are again better behaved in this sense, we showed that the radius of their backwards evolution equals that their forward evolution in \cite{ArrighiUCAUSAL,ArrighiJCSS}. Let us mention that \cite{tHooftCA} embarked on the program of studying how much physics can be recovered within such quantized RCA $\widehat{G}$. 

The early definition \cite{DurrWell,DurrUnitary} of QCA would allow for these non-causal $\widehat{F}$ over ${\cal H}$; and had to be abandoned. However, a characterization of 1D QCA as tensor networks of Matrix Product Unitaries (MPU) \cite{CiracMPU} has recently appeared, which bears strong similarities with the abandoned definition, whilst remaining causal. Indeed, with MPU one looks for a tensor $T^{so}_{is'}$ such that the circular tensor network $$V^{o_0\ldots o_n}_{i_0\ldots i_n}=T^{s_0o_0}_{i_0s_1}T^{s_1o_1}_{i_1s_2}\ldots T^{s_{n}o_n}_{i_ns_0}$$ 
is unitary matrix for any $n$. Interestingly, this is the quantum analogue of the classical bijectivity over ${\cal C}_\infty$ rather than ${\cal C}$, which made bijective CA $G$ an RCA and $\widehat{G}$ a QCA.

\subsection{Historical notes}

Whilst the definition of QCA is now well-established, its beginnings were difficult. The first attempt of a definition would require that state vector of a cell at time $t+1$, should be locally-dependent upon the state of its neighbouring cells at time $t$ \cite{DurrWell,DurrUnitary}. Whilst plausible, we realized that this definition was problematic, as it would still allow for information to propagate at an arbitrary speed \cite{ArrighiLATA,ArrighiIJUC}. The first solid axiomatic definition of QCA was given in \cite{SchumacherWerner}, in terms of $C^*$--algebra, together with a broken proof that these were in fact exactly two-layer quantum circuits, infinitely repeating across space. In \cite{ArrighiLATA,ArrighiIJUC} we rephrased the definition in terms of Hilbert spaces, and clarified the proof as sound for 1D QCA, but gave a counter-example in higher-dimensions taken from \cite{KariCircuit}. In \cite{ArrighiUCAUSAL,ArrighiJCSS} we were able to obtain Th. \ref{th:ucausal}, which we later \cite{ArrighiPQCA} completed into Th. \ref{th:pqca}. Summarizing, QCA are not exactly PQCA, but they are intrinsically simulated by them. 

It took a while, thus, to arrive at the axiomatic definition, and even longer to deduce ways of constructing / enumerating the corresponding instances of this definition. Naturally, this lack of operationality left the gap open for many competing, hands-on definitions of the same concept. A closer examination shows that these competing operational definitions would fall into three classes: the multi-layer quantum circuits \cite{PerezCheung}\cite{ArrighiUCAUSAL}, the two-layer quantum circuits \cite{BrennenWilliams,Karafyllidis,NagajWocjan,Raussendorf,SchumacherWerner,VanDam}, and PQCA \cite{WatrousFOCS,VanDam,InokuchiMizoguchi}. In \cite{ArrighiPQCA} we showed that they all simulate each other in a spacetime-preserving manner, leading us to prefer their simplest, PQCA form.

\section{Universality}\label{sec:universality}

\subsection{Intrinsic universality}\label{subsec:intrinsicuniv}

\begin{figure}[h]
\centering
\includegraphics[scale=0.8, clip=true, trim=0cm 0cm 0cm 0cm]{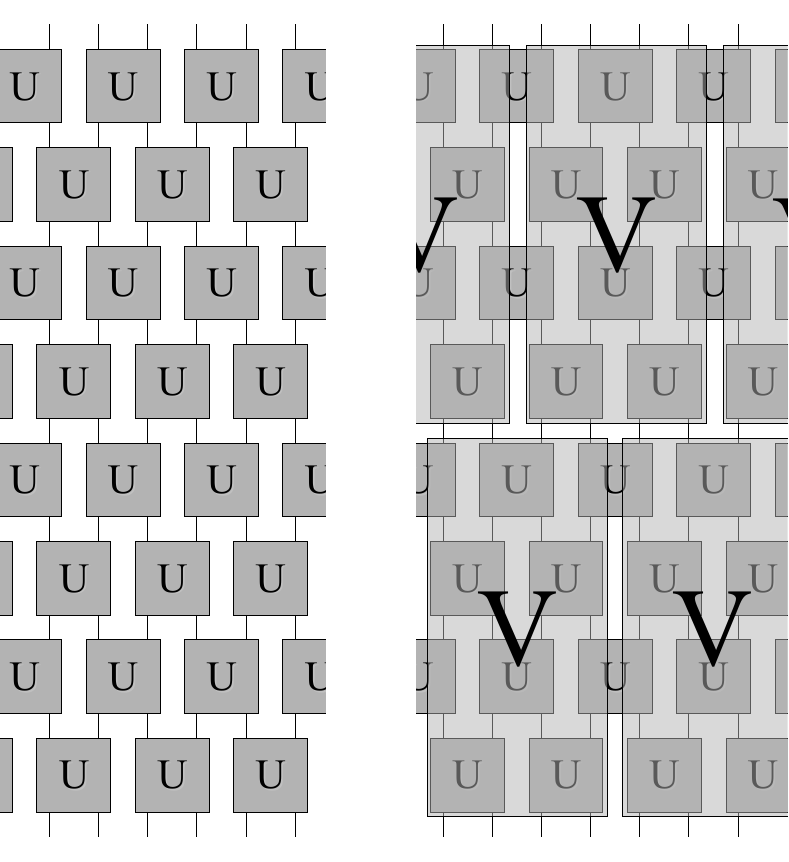}
\hspace{0.7cm}
\raisebox{0.2cm}{\includegraphics[width=3.24cm, clip=true, trim=3.5cm 1.3cm 3.4cm 1.3cm]{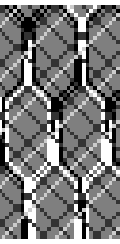}}
\caption{\label{fig:UsimV} \emph{Left:} Partitioned QCA with scattering unitary $U$. Each wire represents a cell, cells are partitioned differently at odd and even time steps. Alternatively one may think of pairs of wires as cells, but then cell positions are translated by a half between each half-step. \emph{Middle:} The PQCA with scattering unitary $U$ intrinsically simulates that with scattering unitary $V$ in four time steps. \emph{Right:} An actual scheme to perform such an intrinsic simulation between 1D PQCA.}
\end{figure}

In Subsection \ref{subsec:structh} we recalled the notion of intrinsic simulation between QCA in order to show that PQCA can simulate all other QCA in a spacetime-preserving manner. Now, once the structure of a model of computation is well-understood, the last step to take in order to try and simplify it even further, is to identify universal instances of the model. Minimal universal instances are particularly useful, as they point towards the threshold physical resources required in order to implement the entire model. Thus, we need to look for a single PQCA which can intrinsically simulate all other PQCA. 

In a PQCA, incoming information gets scattered by a fixed `scattering unitary' $U$, before getting redispatched. We need to find a universal scattering unitary $U$, see Fig. \ref{fig:UsimV}. From a computer architecture point of view, this problem can be recast in terms of finding some fundamental quantum processing unit which is capable of simulating any grid network of quantum processing units, in a space-preserving manner. From a theoretical physics perspective, this is looking for a universal scattering phenomenon, a problem which we could phrase in humorous form as: ``A physicist is taken on a desert Island where he is allowed only one type of elementary particle. Which one would he choose, whose scattering behaviour is rich enough so that it can simulate all the others?''.

\begin{figure}
 {\centering
\hfill
\begin{minipage}[t]{4cm}
\includegraphics[scale=.28, clip=true]{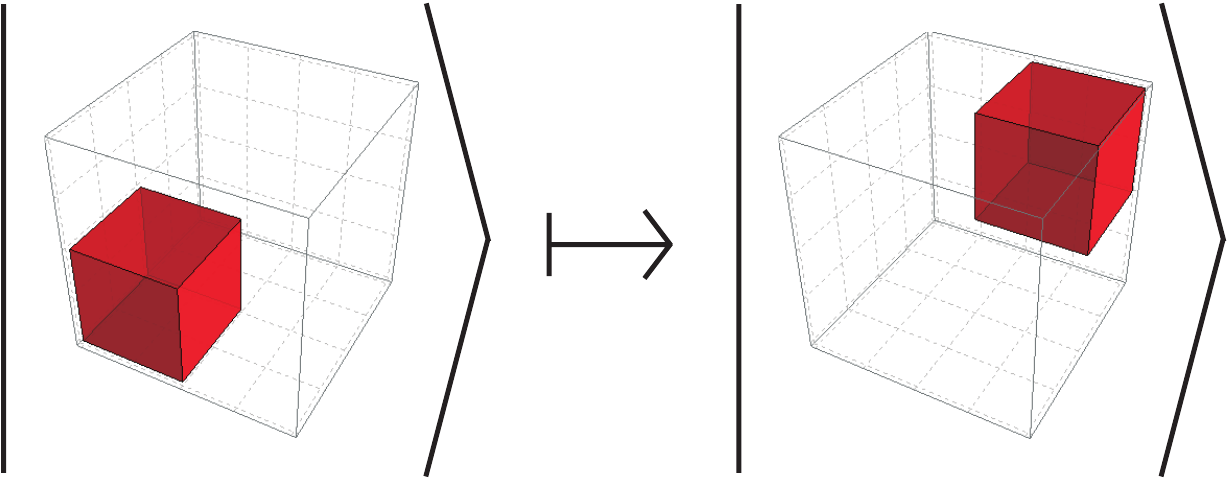}\\
{\small A signal alone moves across space.}
\end{minipage}
\hfill
\begin{minipage}[t]{4cm}
\includegraphics[scale=.28, clip=true]{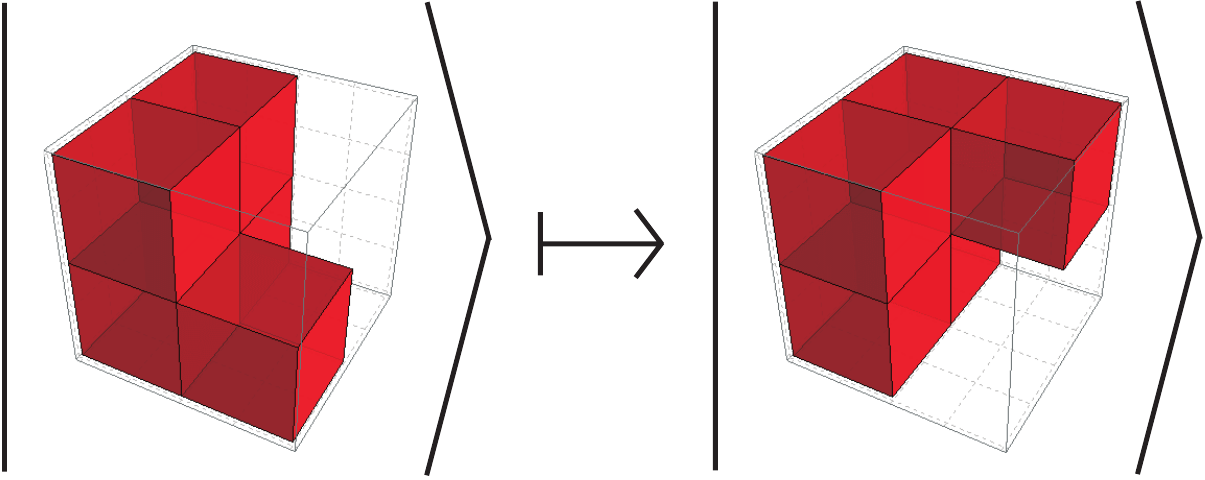}\\
{\small A signal which collides straight into a wall ``bounces'' off the wall.}
\end{minipage}
\hfill~\\
\hfill
\begin{minipage}[t]{4cm}
\includegraphics[scale=.24, clip=true]{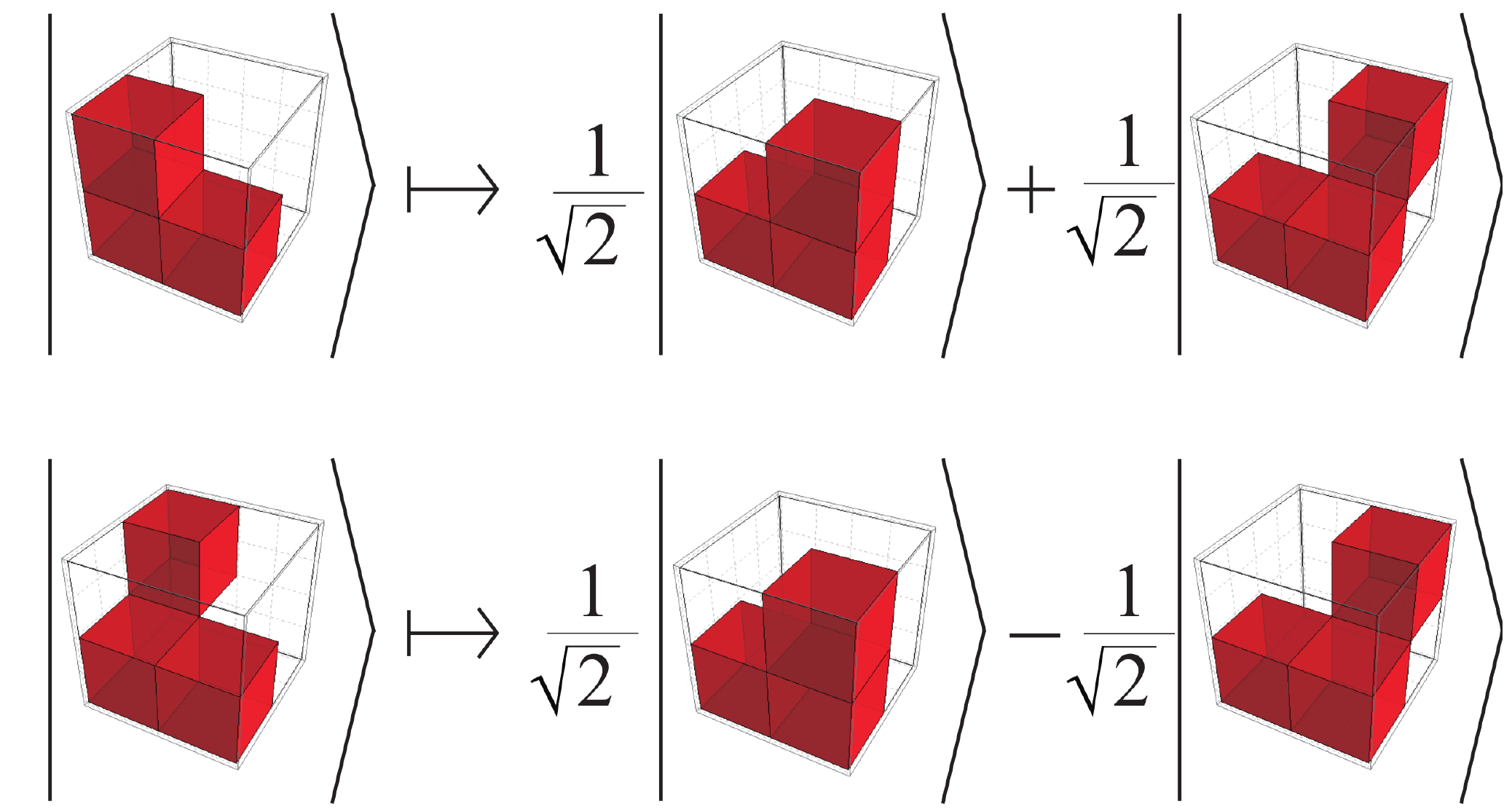}\\
{\small A signal which collides at an angle produces a superposition as in the Hadamard gate.}
\end{minipage}
\hfill
\begin{minipage}[t]{4cm}
\includegraphics[scale=.28, clip=true]{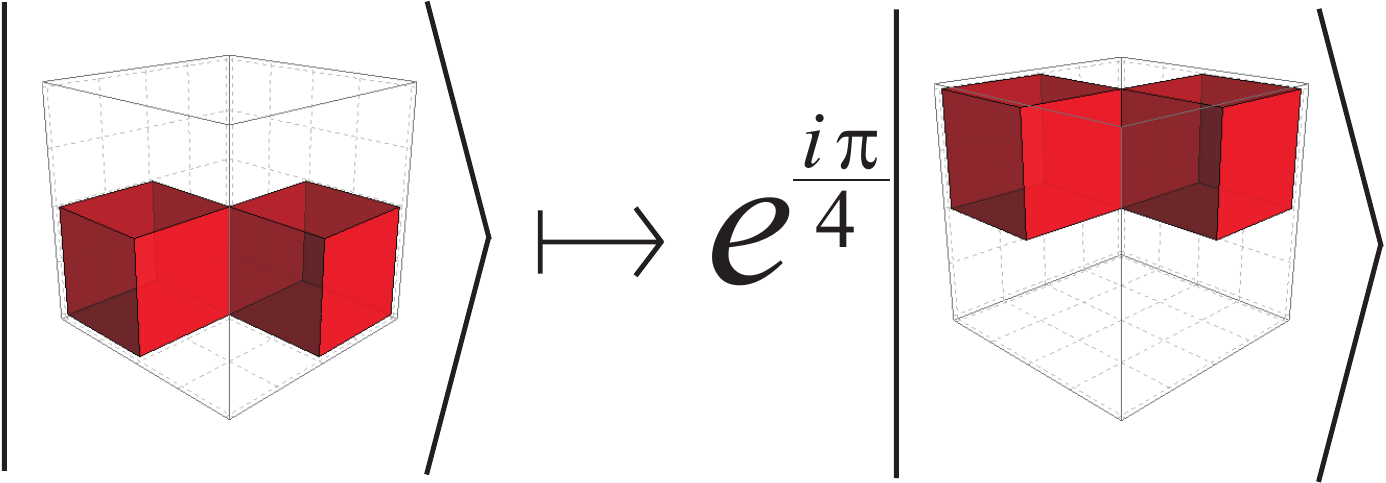}\\
{\small When two signals cross each other diagonally, a complex phase is added as in the controlled-($\frac{\pi}{4}$) gate.}
\end{minipage}
\hfill~
\\\vspace{5pt}
\caption{{\em A minimal intrinsically universal 3D PQCA}: the scattering unitary of the ``Quantum Game of Life''. \label{fig:QGOL}}}
\end{figure}

We began the search for an intrinsically universal PQCA in dimension $1$, which is feasible \cite{ArrighiDCM,ArrighiFI} (see also Fig. \ref{fig:UsimV}) but difficult, because wires cannot cross over. We then then tackled the problem in the general $n$-dimensional case, where we could find a much simpler solution \cite{ArrighiSimple,ArrighiNUQCA}. Eventually we reached a minimal, $3$-dimensional construction \cite{ArrighiQGOL}. This so-called `Quantum Game of Life' roughly works as follows. Each cell contains just one qubit---in Fig. \ref{fig:QGOL} the cell is represented as little cube, red if the qubit is in state $\ket{1}$, transparent if it is in state $\ket{0}$. At even steps, the cube of 8 qubits at cells $\{0,1\}^3$, as well as all of its $(2{\mathbb{Z}})^3$--translates, each undergo a scattering unitary $U$, synchronously. At odd steps, the cube of 8 qubits at cells $\{1,2\}^3$, as well as all of its $(2{\mathbb{Z}})^3$--translates, each undergo $U$ again, synchronously. The scattering unitary $U$ is given in Fig. \ref{fig:QGOL}, by means of its action over a small number of basis states (the full definition follows by linear extension and assuming rotation-invariance).\\
Observe that when there is just one red, e.g. at the left-bottom-front corner of the cube, it just moves across the cube. But because of the staggering between even and odd steps, it will find itself at the left-bottom-front corner of the new cube and again move across---this is the mechanism whereby signals are made. We also need to be able to redirect our signals, so we need red walls to form stable patterns, and demand that if a fifth red comes along, it bounces off. The Hadamard is implemented as a special case of this deflection : if the signal bounces on a edge, it skids or bounces, in a quantum superposition. Finally a two qubit interaction happens when two signals cross, and a phase gets added.\\

\begin{figure}
\centering
\includegraphics[height=3.6cm, clip=true, trim=0cm 0cm 0cm 0cm]{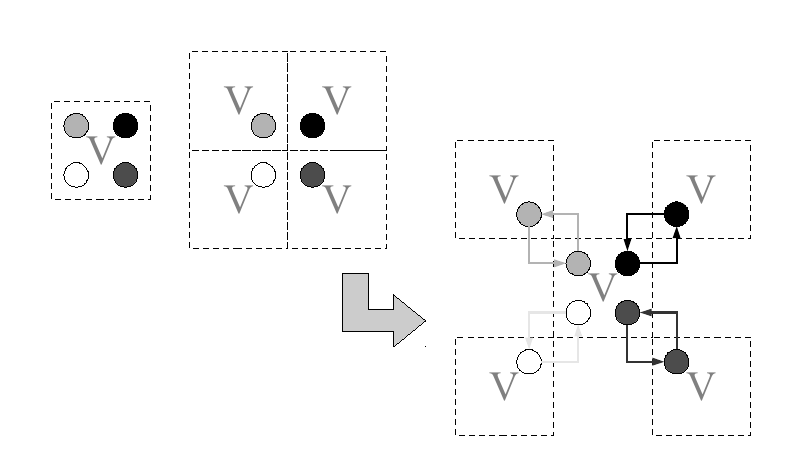}
\raisebox{-0.05cm}{\includegraphics[height=3.6cm, clip=true, trim=0cm 0cm 0cm 0cm]{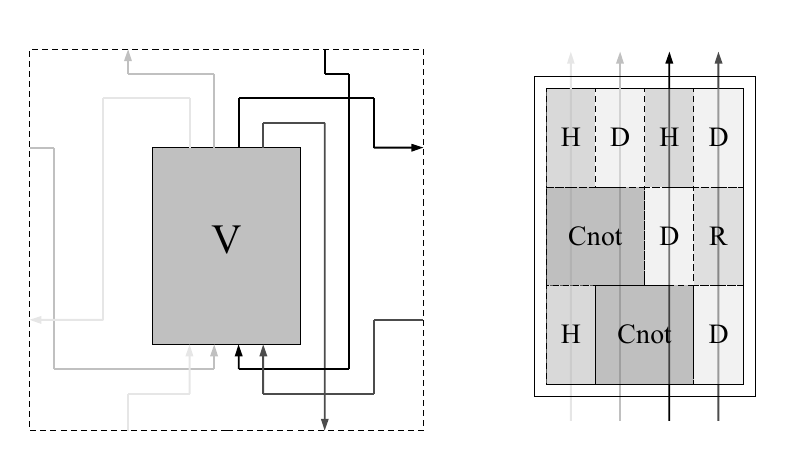}}
\\
\includegraphics[height=1.75cm, clip=true, trim=0cm 0cm 0cm 0cm]{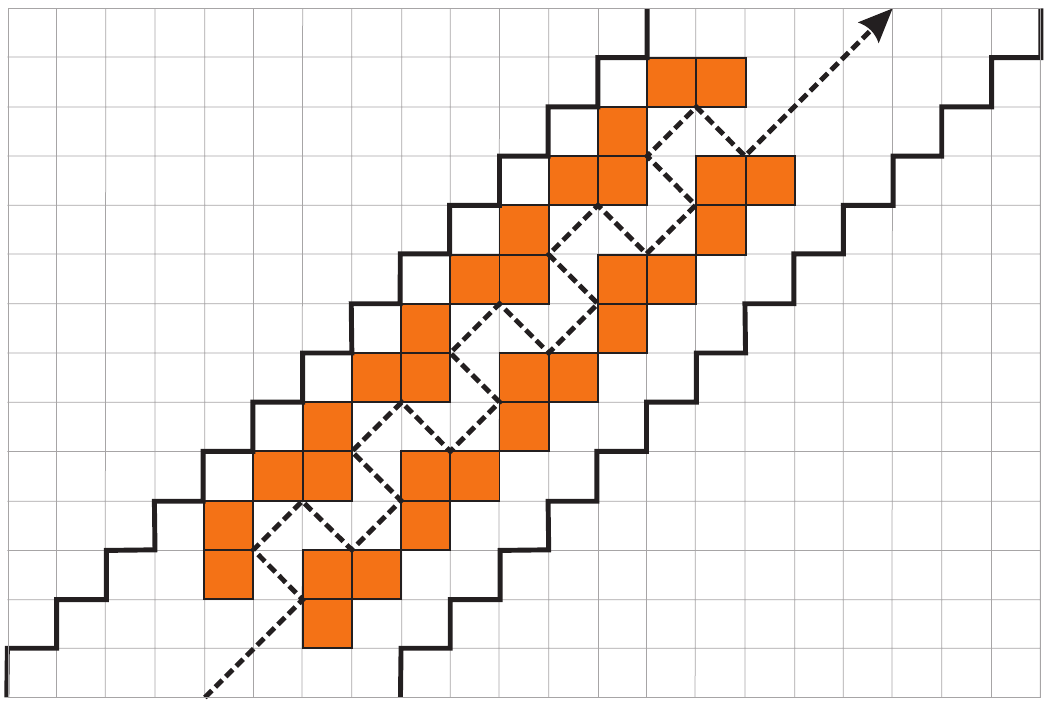}
\includegraphics[height=1.75cm, clip=true, trim=0cm 0cm 0cm 0cm]{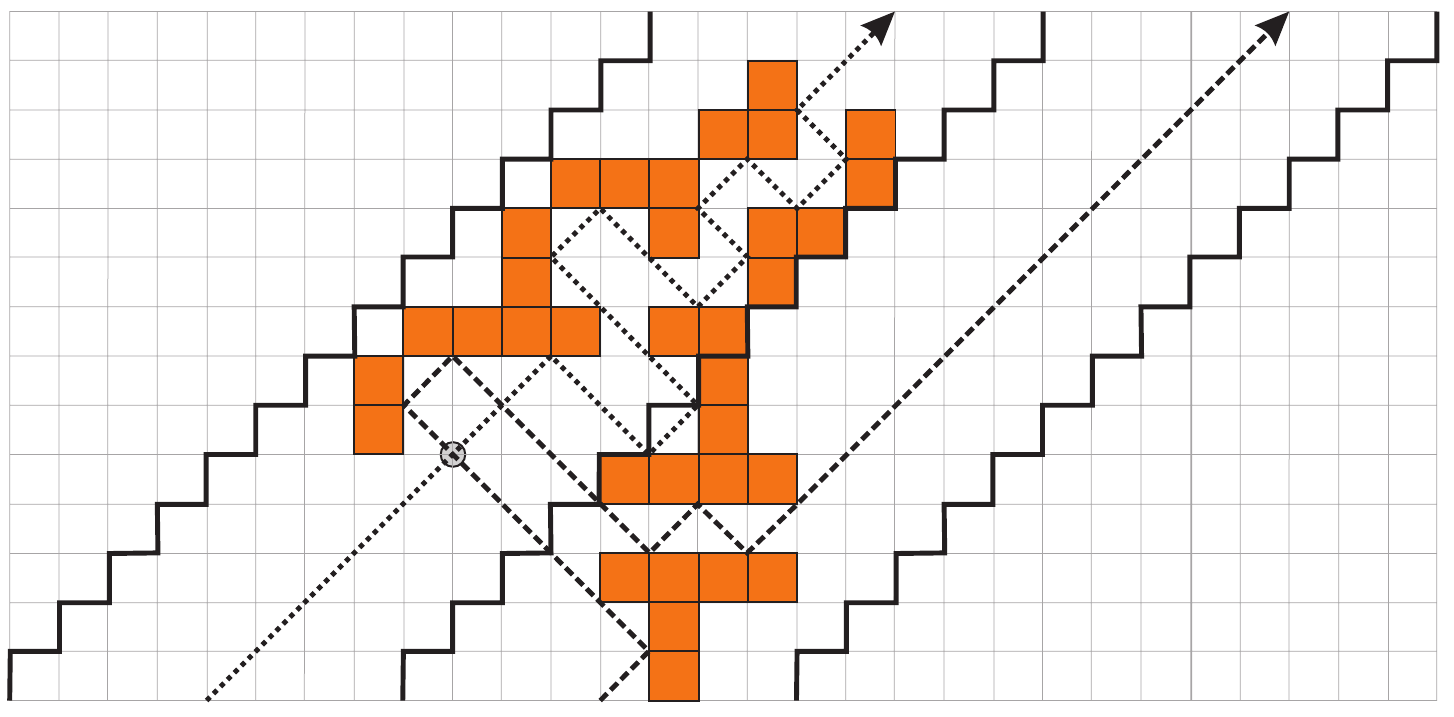}
\includegraphics[height=1.75cm, clip=true, trim=0cm 0cm 0cm 0cm]{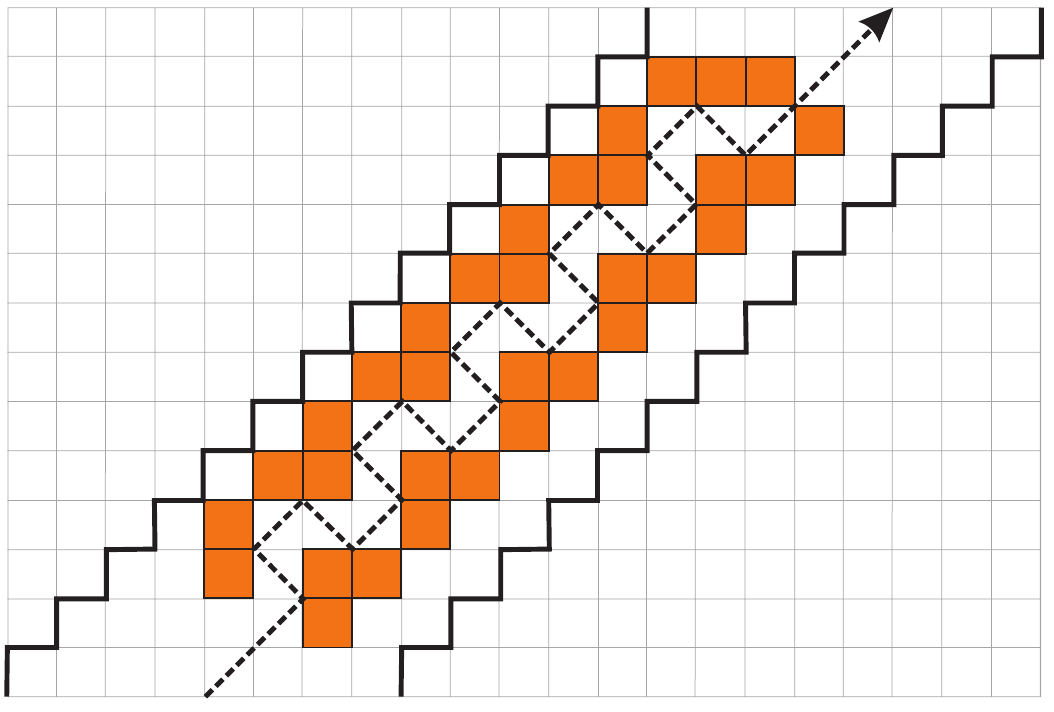}
\includegraphics[height=1.75cm, clip=true, trim=0cm 0cm 0cm 0cm]{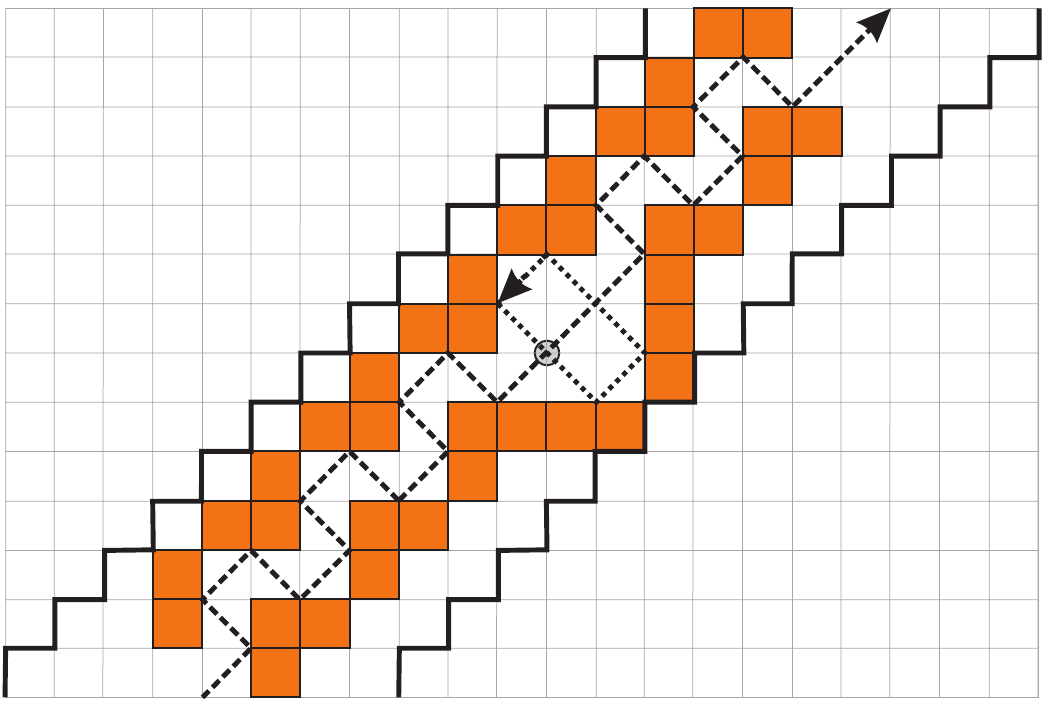}
\caption{\label{fig:Flattening} {\em Intrinsic simulation of a 2D PQCA by another.} {\em Top left:} The partitions at odd and even steps of the PQCA with scattering unitary $V$ overlap, but in the simulating PQCA they will be laid out side-by-side. {\em Top right:} The simulating PQCA emulates each $V$ in parallel, as well as the wirings between the $V$. {\em Bottom:} It does so by combining tiles (of fixed shapes and taking a fixed number of steps to be traversed, implementing universal quantum gates) into a layout that implements $V$.}
\end{figure}

The outline of the proof that this PQCA is intrinsically universal is as follows (see Fig. \ref{fig:Flattening} for a 2D illustration of this argument). First make fixed-shaped tiles, each implementing one of the universal quantum gates of the quantum circuit model in a fixed number of time steps, out of these walls and signals. Next, combine these tiles into a layout that implements $V$, the scattering unitary of the PQCA to be simulated. Repeat this layout across space. Finally, plug the outputs of each simulated $V$ gates into the inputs of the neighbouring ones, so that they feed each other, thereby implementing the staggered structure of the simulated PQCA.
  
Notice that the hereby constructed scattering unitary is over 8 qubits, which is much more complicated that the 2 qubits gate sets that are universal for quantum circuits. This is because simulation of a QCA $G$ has to be done in a parallel, spacetime-preserving manner, and because we must simulate not just one iteration of $G$ but several ($G^2$, $G^3$\ldots, i.e. after every iteration we must get ready for the next one). Thus intrinsic universality is a much more stringent requirement than quantum circuit universality.

\subsection{Other kinds  of universality}

\paragraph{Quantum Turing machine.} We just constructed a PQCA that is capable of simulating any other PQCA and hence any QCA.  But does it mean that this PQCA is capable of running any quantum algorithm? Clearly, the question amounts to whether QCA are universal for QC. In the sense of the quantum circuit model the answer is clearly yes, simply by inspection of the construction of Subsection \ref{subsec:intrinsicuniv}, Fig. \ref{fig:QGOL} and \ref{fig:Flattening} in particular. In the sense of the quantum Turing machine, the answer is clearly yes also, as was proven in \cite{WatrousFOCS}. In this construction, a QCA with alphabet $\Sigma\times \{0,1\} \times S$ simulates a Quantum Turing machine with alphabet $\Sigma$ and internal states $S$. The way this works is that each cell is capable of hosting the head of the Turing machine : it has enough state space to store both the symbol at this location of the tape; the internal state of the head; and whether the head is actually there or not. 

\paragraph{Quantum circuit universality.} The so-called ``physical universality'' is specific-kind of quantum circuit universality for QCA, more stringent than the early work of \cite{VanDamUniversal}. Indeed, in all the above-mentioned universality constructions, part of the state space $\Sigma$ of each cell is used to encode `the program' (i.e. what dynamics is to be simulated), whilst the other part is used to encode `the data' (i.e. the states whose evolution is being simulated). One may demand that this is not the case, and wish to have a convex region $X$ in which the data (the input to a quantum circuit) lies untouched, without any preparation, whereas only the surroundings $\overline{X}$ are allowed to code for the program (the quantum circuit to be applied). The requirement is that after a precise number of time steps, the data is to be found at the very same place, evolved according to the specified quantum circuit. Such a construction is achieved in \cite{SchaefferPhysUniv}. In this construction the data within $X$ is left to ``explode'' into $\overline{X}$, where it gets treated and redirected towards $X$.

\paragraph{Computability.} The (strong) physical Church-Turing thesis states that ``any function that can be (efficiently) computed by a physical system can be (efficiently) computed by a Turing machine''. Because there are concrete examples of functions that cannot be computed by Turing machines (e.g. the famous halting function $h:{\mathbb N}\to \{ 0,1 \}$), the physical Church-Turing thesis makes a strong statement about physics' (in)ability to compute. 

The discovery of QC algorithms has shaken the strong version of the thesis. But what about the original version---could it be that a QC might compute functions that were not computable classically? Quantum theory imposes that physical systems evolve unitarily : according to a unitary matrix when the system is finite-dimensional; according to a unitary operator otherwise. It follows that finite quantum circuits can always be simulated (very inefficiently) on a classical computer just via matrix multiplications. Therefore these do not endanger the original version of the thesis. Yet nothing forbids \cite{NielsenComputability,Kieu,NielsenMore} that unitary operators, on the other hand, break the thesis, e.g. $U=\sum\ket{i,h(i)\oplus b}\bra{i,b}$. That is unless the limitation comes from other physical principles. 

That physically motivated limitations lead to the physical Church-Turing thesis 
was already argued in \cite{Gandy} by Gandy, Turing's former PhD student. There, the main idea is that causality (i.e. bounded velocity of propagation of information) together with homogeneity (i.e. the rules of physics are the same everywhere and everywhen) and finite density (i.e. bounded number of bits per volume) entail CA-like evolutions, which are computable. Actually the proof relies upon a few more assumptions (an euclidean-like space, whose state can be described piecewise) and a delicate formalism. But the main issue with this proof of the physical Church-Turing thesis based upon physics principles, is that it complete ignores quantum theory. Quantum theory demands that the bounded density principle be updated, changing the word `bits' to  the word `qubits'. 

When we do so, the updated set of principles entail QCA-like evolutions in the axiomatic-style of Def. \ref{def:qca}. We can then apply Th. \ref{th:ucausal}. Armed with a robust notion of computability upon vector spaces \cite{ArrighiCIE}, we were able show that these are computable \cite{ArrighiGANDY}. This provided a proof of the Church-Turing thesis based upon quantum physics principles. 

\section{Simulation}\label{sec:simulation}

Let us take a step back to realize how the previous results chain up. In Section \ref{sec:structure} we showed that discrete-space discrete-time quantum theory, i.e. QCA, can be intrinsically simulated by PQCA. In Section \ref{sec:universality} we constructed a minimal, intrinsically universal PQCA. Logically, this entails that any lattice discrete-space discrete-time quantum physics phenomenon can be expressed within this particular PQCA.\\
Let us evaluate whether such statements are applicable in practice. Let us pick up one of the simplest and most fundamental physics phenomenon, namely the free propagation of the electron, and see whether it can be re-expressed by means of some simple PQCA.

\subsection{The Dirac QCA}\label{subsec:DiracQCA}

\begin{figure}
{\centering
\includegraphics[width=5cm]{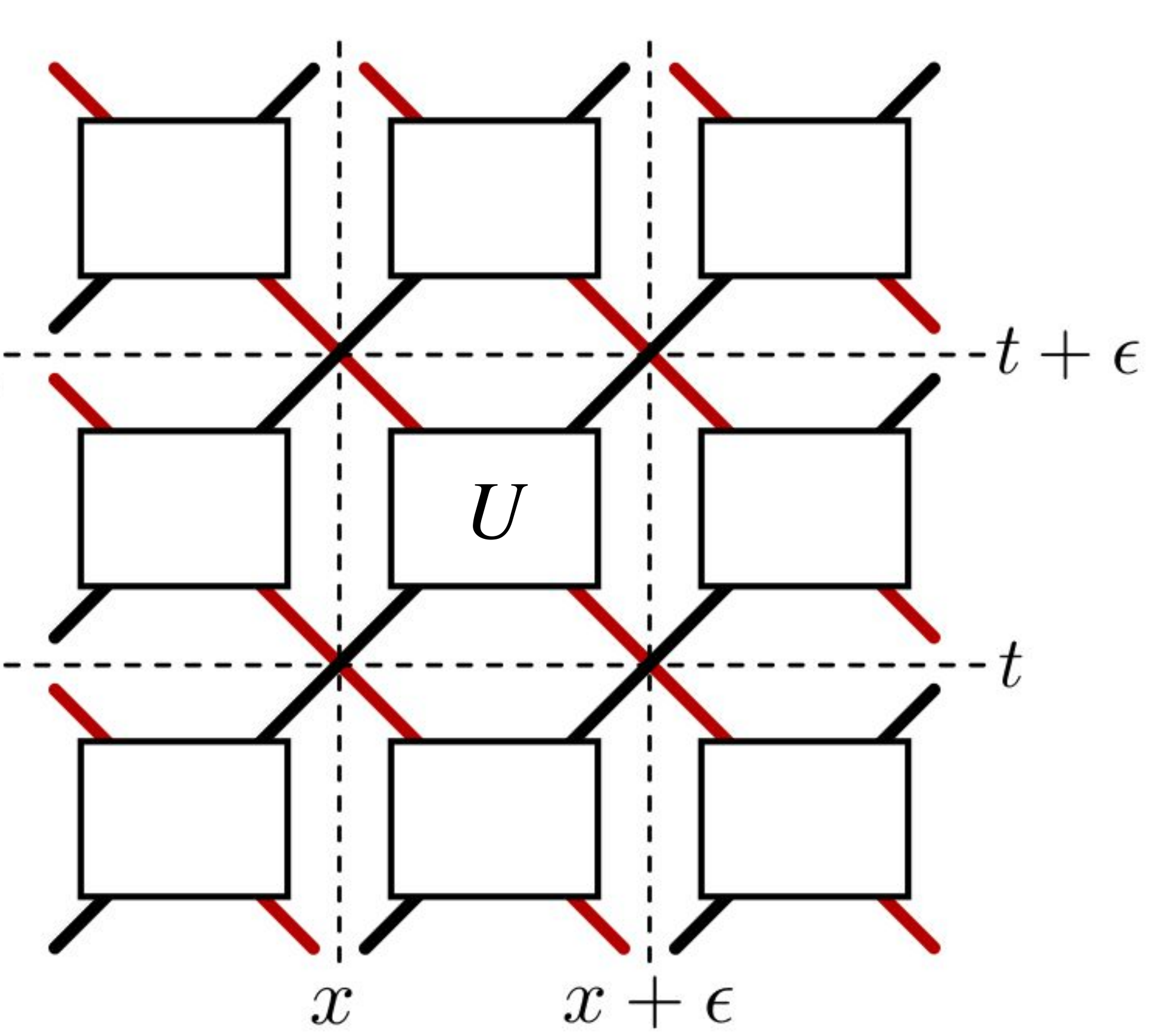}
\caption{\label{fig:DiracQCA}{\em The Dirac QCA:} each space-time point $(x,t)$ can be occupied by a left and a right-moving particle. All gates are identical and given by the matrix $U$, which lets a particle change direction with an amplitude that depends on its mass. 
}
}
\end{figure}

The equation governing the free propagation of an electron is called the Dirac equation. Let us describe a PQCA model for it, the so-called Dirac QCA. For this Dirac QCA we adopt the conventions depicted in Fig.~\ref{fig:DiracQCA}. Each red (resp. black) wire carries a qubit, which codes for the presence or absence of a left-moving (resp. right-moving) electron. Thus there can be at most two electrons per site $x = \varepsilon k$, $k \in \mathbb Z$ (where the red and black wires cross).

The scattering unitary of this QCA is given by
\begin{align*}
&U= \left(
\begin{array}{cccc}
1 & 0 & 0 & 0\\
0 & -\ii s & c & 0\\
0 & c & -\ii s & 0\\
0 & 0 & 0 & 1
\end{array}
\right)\\
&=1\oplus \sigma_1 \exp(-\ii m\varepsilon \sigma_1) \oplus 1
\end{align*}
with $c=\cos(m \varepsilon)$ and $s=\sin(m \varepsilon)$, where $m$ stands for the mass of the electron, and $\varepsilon$ the spacetime discretization step, and the Pauli matrices as usual: 
\begin{align*}
&\sigma_0=\left(\begin{array}{cc}
 1 &0 \\
 0 &1\end{array}\right),\quad\sigma_1=\left(\begin{array}{cc}
 0 &1 \\
 1 &0\end{array}\right),\quad\sigma_2=\left(\begin{array}{cc}
 0 &-i \\
 i &0\end{array}\right),\quad\sigma_3=\left(\begin{array}{cc}
 1 &0 \\
 0 &-1\end{array}\right).
\end{align*}
Notice that the components are ordered so that when the mass is zero, the particles do not change direction, i.e. a right-moving electron is transfered from $x$ to $x+\varepsilon$, etc.


To see whether this QCA implements the Dirac equation, let us consider the one-particle sector, i.e. restrict to the QCA to the subspace spanned by states of the form $\ket{\ldots 00100 \ldots}$. Let $\psi^{-}(t,x)$ (resp. $\psi^{+}(t,x)$) be the amplitude of that single particle being on a red (resp. black) wire at their intersection at point $(x,t)$. We have
\begin{equation*}\label{eq:sp}
\begin{split}
\psi^+(t+\varepsilon,x)&=c\psi^+(t,x-\varepsilon)-\ii s\psi^-(t,x)\\ 
\psi^-(t+\varepsilon,x)&=c\psi^-(t,x+\varepsilon)-\ii s\psi^+(t,x).
\end{split}
\end{equation*}
Expanding to the first order in $\varepsilon$ gives
\begin{align*}
\varepsilon\partial_t \psi^+&=-\varepsilon\partial_x \psi^+-\ii m\varepsilon\psi^-\\ 
\varepsilon\partial_t \psi^-&=+\varepsilon\partial_x \psi^--\ii m\varepsilon\psi^+\\
\partial_t \psi&=-\sigma_3\partial_x\psi -\ii m\psi\label{eq:Dirac}
\end{align*}
with $\psi=(\psi^+,\psi^-)^T$. We recognize the $(1+1)$--dimensions Dirac equation governing the free propagation of an electron of mass $m$. The upper (resp. lower) component of the vector gets transported right (resp. left), but the mass mixes these components.

\subsection{Further simulation results}\label{subsec:FurtherSimulationResults}

\paragraph{The one-particle sector : quantum walks.} Subsection \ref{subsec:DiracQCA}, illustrated how to restrict QCA to their `one-particle sector', i.e. to configurations of the form $\ket{\ldots 00s00 \ldots}$ with $s\in S=\Sigma\setminus\{0\}$. Each of these represents a single `particle', standing at a some position $x$ with internal state $s$. The Hilbert space of superpositions of these configurations is a rather small subspace of ${\cal H}$, which is better described as ${\cal H}_{\mathbb{Z}}\otimes {\cal H}_{S}$---i.e. superpositions of position-state pairs $\ket{x,s}$. The amplitude of $\ket{x,s}$ is usually written $\psi^s(x)$ and one needs $\sum_{x,s} |\psi^s(x)|^2=1$. This one-particle sector of QCA has a life of its own. It is the playground for a huge field of research known by the name of Quantum Walks (QW). A QW, therefore, is essentially an operator driving the evolution of a single particle on the lattice, through local unitaries.

One reason for the popularity of QW is that a whole series of novel Quantum Computing algorithms,
for the future Quantum Computers, have been discovered via QW, e.g. \cite{BooleanEvalQW,ConductivityQW}, or are better expressed using QW, e.g the Grover search. In these QW-based algorithms, however, the walker usually explores a graph, which is encoding the instance of the problem, rather than a fixed lattice. No continuum limit is taken.

The focus here will remain with the other reason, which is the ability of QW to simulate certain quantum physical phenomena, in the continuum limit---thereby providing:
\begin{itemize}
\item Simple discrete toy models of physical phenomena, that conserve most symmetries
(unitarity, homogeneity, causality, sometimes even Lorentz-covariance)---thereby providing
playgrounds to discuss foundational questions in Physics.
\item Quantum simulation schemes, for the near-future simulation devices, in the way that was envisioned by Feynman when he invented QC \cite{FeynmanQC,FeynmanQCA}.
\item Stable numerical schemes, even for classical computers---thereby guaranteeing convergence
as soon as they are consistent.
\end{itemize}

Subsection \ref{subsec:DiracQCA} is a simplified presentation of the original arguments by \cite{BenziSucci,Bialynicki-Birula,MeyerQLGI} suggesting that QW can simulate the Dirac equation. In \cite{ArrighiDirac} we gave a rigorous proof of convergence, given regular enough initial conditions, including in $(3+1)$--dimensions---without the need to actually solve the QW evolution as was done in \cite{StrauchPhenomena}. The zitterbewegung effect is discussed in \cite{ChandrashekarDiracQW}. An axiomatic derivation of these schemes is given in \cite{DAriano,DAriano3D,RaynalDirac}. We discussed conservation of symmetries, including Lorentz-covariance in 
\cite{ArrighiLorentzCovariance}, and so did \cite{PaviaLORENTZ,PaviaLORENTZ2,DebbaschLORENTZ}. The Klein-Gordon equation can also be simulated via this QW once the appropriate decoupling is performed, as explained in \cite{IndiansDirac,MolfettaDebbasch} and in our subsequent generalization \cite{ArrighiKG}. The Schr\"odinger equation can be obtained in a similar fashion, but by scaling space and time differently, i.e. $\Delta_x=\varepsilon$ but $\Delta_t=\varepsilon^2$, see \cite{StrauchCTQW,StrauchShrodinger,BoghosianTaylor2}.

Once it was realized that QW could simulate free particles, the focus shifted towards simulating particles in some background field \cite{cedzich2013propagation,di2014quantum,di2016quantum,arnault2016quantum,marquez2017fermion}, by means of non-translation-invariant QW. The question of the impact, of these inhomogeneous fields, upon the propagation of the walker gave rise to lattice models of Anderson localization \cite{WernerLocalization,JoyeLocalization}. Surprisingly, they even gave rise to lattice models of particles propagating in curved spacetime \cite{MolfettaDebbasch2014Curved,DebbaschWaves,ChandrashekarSSCurved}, see also our generalizations \cite{ArrighiGRDirac,ArrighiGRDirac3D}.
 
\paragraph{The many-particle sector.} Recently, the two-particle sector of QCA was investigated from a quantum simulation perspective, with the two walkers interacting via a phase (similar to the Thirring model \cite{DdV87}). This was shown to produce molecular binding between the particles \cite{ahlbrecht2012molecular,PaviaMolecular}. In the many-particle sector, the problem of defining a concrete QCA that would simulate a specific interacting quantum field theory (QFT) had remained out of reach until \cite{ArrighiQED}.  In this paper, we were able to give a first QCA description of quantum electrodynamics (QED) in $(1+1)$--dimensions (a.k.a the Schwinger model).  

\paragraph{Trotterization of a nearest-neighour Hamiltonian.} QCA are in discrete-space and discrete-time. Let us consider their cousins in discrete-space but continuous-time, i.e. lattices of quantum systems interacting according to a nearest-neighbour translation-invariant Hamiltonian. These are very common in Physics e.g. in condensed matter or statistical quantum mechanics (spin chains, Ising models, Hubbard models\ldots), or towards QC (as candidate architectures \cite{FitzsimonsTwamley,Benjamin,Twamley,WeinsteinHellberg}, for quantum information transport \cite{Bose}, for  entanglement creation \cite{Subrahmanyam,SubrahmanyamLakshminarayan,BrennenWilliams}, as universal QC \cite{VollbrechtCirac,NagajWocjan}\ldots).\\
Up to groupings and reencodings, focussing here on the 1D case just for simplicity, nearest-neighbour translation-invariant Hamiltonians work as follows \cite{VollbrechtCirac}. A global, continuous-time evolution $G(t)$ is induced, by giving a hermitian matrix $h$ over $\mathcal{H}_{\Sigma}\otimes \mathcal{H}_{\Sigma}$ verifying that $h\ket{00}=0$, according to $$G(t)=e^{-iHt}=\sum_n \frac{(-iHt)^n}{n!}$$ 
with $H=\sum_x h_x$ where $h_x$ stands for $h$ as acting over positions $x$ and $x+1$.\\
From a practical implementation point-of-view, nearest-neighbour Hamiltonians are central, and will be for a long time. Indeed, although there are a number of remarkable exceptions \cite{WernerElectricQW,Alberti2QW,Sciarrino}, most the leading-edge lattice-based quantum simulation devices remain better described as continuous-time evolutions \cite{Bloch}. From a theoretical point-of-view, however, nearest-neighbour Hamiltonians suffer the same downsides as the rest of non-relativistic quantum mechanics to which they pertain. Namely, strictly speaking they do allow for superluminal-signalling, hopefully in some negligible, exponentially tailing off manner, relying upon some Lieb-Robinson type of argument that can sometimes go wrong \cite{EisertSupersonic}. Intuitively, this is because even though it is the case that in an infinitesimal $\delta_t$ of time a cell only interacts with its neighbour, this is no longer true after any finite period of time $\Delta_t$, however small, as $G(\Delta_t)$ includes terms of the form $\prod_x h_x$ and $[h_x,h_{x+1}]\neq 0$ if information is to propagate at all.  

Still, there is a strong connection between nearest-neighbours Hamiltonians and QCA, which arises from the Trotter-Kato formula (a.k.a Baker-Campbell-Thomson or operator-splitting method): 
$$e^{i\Delta_t (H_o+H_e)}= e^{i\Delta_t H_o}e^{i\Delta_t H_e}+O(\Delta_t^2).$$
Indeed, let $H_e=\sum_{x\in2\mathbb{Z}} h_x$ and $H_o=\sum_{x\in2\mathbb{Z}+1} h_x$, and readily get that $G(\Delta t)\approx G$, where $G$ is the PQCA induced by the scattering unitary $U=e^{i\Delta_t h}$---back in discrete-space discrete-time.\\ Are space and time back on an equal footing, thanks to this `trotterization procedure'? Not quite. For this approximation to hold mathematically, one still needs that $\Delta_t\ll\Delta_x$---for instance by setting $\Delta_x=\varepsilon$ and $\Delta_t=\varepsilon^2$ as was done earlier in order to get the non-relativistic, Schr\"odinger equation. Thus, QCA arising by trotterizing nearest-neighbour Hamiltonians are generally non-relativistic models \cite{ArrighiChiral}, unless they are carefully engineered otherwise, as we did in \cite{ArrighiUnifiedQW}. Still, this is not the only use of the Trotter-Kato formula, which has turned out to be an ubiquitous mathematical tool in this field. 

\paragraph{Noise, thermodynamics.} To the best of our knowledge there has not been much studies of noisy QCA, i.e. replacing unitary operators by quantum operators (a.k.a TPCP maps), with a handful of exceptions \cite{AvalleNoisyQCA}, in the many-particle sector. The one-particle sector has been thoroughly studied on the other hand, e.g. studying the transition for QW (ballistic transport) to random walks (diffusion) \cite{LoveBoghosian}. In \cite{ArrighiNoisyQW} we studied this transition in parallel with the continuum limit to PDE, where the Dirac equation turns into a Lindblad equation and then a telegraph equation---making the argument that noisy quantum simulation devices can still be useful, to simulate noisy quantum systems.\\ 
The connection to toy models of thermodynanics is also a promising one. In \cite{RomanelliTemperature} the thermalization of a QW is observed. In the many-particle sector, Clifford QCA \cite{SchlingemannWerner,GutschowCLIFFORD} (a subcase of QCA which can be classically simulated) were shown to produce fractal pictures \cite{NesmeGutschow} and then gliders \cite{GUWZ}, used to show bounds on entanglement propagation and von Neumann entropy creation in QCA. These hands-on toy models may eventually bring about interesting, complementary point-of-view on the blossoming field of quantum thermodynamics, by taking space into account, which is believed to be a key ingredient of the quantum-to-classical transition \cite{PazZurek}. 

\section{Conclusion}\label{sec:conclusion}

{\em Summary.} Quantum cellular automata (QCA) are quantum evolutions of lattices of quantum systems, as resulting from nearest neighbour-interactions. This sentence, however, could be understood in many ways:
\begin{itemize}
\item In terms of finite-depth circuits of local quantum quantum gates, infinitely repeating across space---amongst the various shapes of circuits proposed, it is now known that the simplest, namely Partitioned QCA (PQCA), can simulate all the others.
\item In continuous-time in terms of sum of local Hamiltonians. It is now clear that integrating these over a small period of time, yields a discrete-time evolution that can again be simulated by a PQCA.
\item In more abstract terms, as the axiomatic requirement that the evolved state vector; or quasi-local algebra; or density matrix, be locally dependent. Locally dependent state vectors have turned out to  make little sense, but the last two were shown to be equivalent and ultimately again simulated by PQCA.
\end{itemize}
Amongst PQCA, some instances were shown to simulate the quantum Turing machine; the free electron; the electron in an electromagnetic field; the electron in curved spacetime, including in $(3+1)$--dimensions. Lately some were shown to model interacting quantum field theories as PQCA, namely the Thirring model, and QED in $(1+1)$--dimensions. Ultimately, some particular instances were shown to simulate all other instances. 

{\em Perspectives.} I would love to see QCA come true, implemented in the labs. This is certainly a fascinating topic in which cold atoms \cite{WernerElectricQW,Alberti2QW,Bloch}, integrated fiber optics \cite{Sciarrino} and hopefully superconducting qubits \cite{SuperconductingQSim} will have a say. As a theoretician it would be adventurous for me to comment much on this perspective. It seems quite likely however that noisy implementations will see the light in the next ten years, with progressive improvements from there. At least nothing, at theoretical level, prevents it. The fact that much physics phenomena can be cast as QCA is an encouragement in this sense; it suggests that physics might naturally implement QCA, at its fundamental level. \\
But exactly, how much particle physics can be recast as QCA? Will QCA provide us with an alternative mathematical framework for interacting quantum field theories? Hopefully a clearer one, more explanatory, readily providing us with quantum simulation algorithms to draw predictions? These questions are, at the theoretical level, the obvious and most likely continuation of the trend of work which I presented in this overview. I am very optimistic about them : my personal belief is that these will be answered positively, probably within the next ten years. Of course I foresee many technical difficulties along the way, but no good reason why this could not be done. Thus, I wish to take this opportunity to encourage young researchers to engage these noble and realistic aims, hopefully enjoying the same collaborative spirit that has reigned over this research community in the last decade.\\
I do not believe, however, that QCA can account for General Relativity. Nor do I believe that they constitute the ultimate model of distributed Quantum Computation. In both cases, an ingredient is missing : the ability to depart from the grid and make the topology dynamical. Quantum Causal Graph Dynamics \cite{ArrighiQCGD} are unitary operators over quantum superpositions of graphs. The graphs constrain the evolution by telling whom can interact with whom ; but at the same time they are the subject of the evolution, as they may vary in time. The possible connections between this further generalization of QCA, and Quantum Gravity, are intriguing.

\section*{Acknowledgements} I was lucky to have, as regular co-authors, great researchers such as Pablo Arnault, C\'edric B\'eny, Gilles Dowek, Giuseppe Di Molfetta, Stefano Facchini, Terry Farrelly, Marcelo Forets, Jon Grattage, Iv\'an M\'arquez, Vincent Nesme, Armando P\'eres, Zizhu Wang, Reinhard Werner. I would like to thank Jarkko Kari and Grzegorz Rozenberg for inviting me to write this overview, a task which I had been postponing for too long.

\bibliography{biblio}
\bibliographystyle{plain}

\end{document}